\documentclass[sigconf, nonacm]{aamas}

\usepackage{balance} % for balancing columns on the final page

\makeatletter
\gdef\@copyrightpermission{
  \begin{minipage}{0.2\columnwidth}
   \href{https://creativecommons.org/licenses/by/4.0/}{\includegraphics[width=0.90\textwidth]{by}}
  \end{minipage}\hfill
  \begin{minipage}{0.8\columnwidth}
   \href{https://creativecommons.org/licenses/by/4.0/}{This work is licensed under a Creative Commons Attribution International 4.0 License.}
  \end{minipage}
  \vspace{5pt}
}
\makeatother

\setcopyright{ifaamas}
\acmConference[AAMAS '25]{Proc.\@ of the 24th International Conference
on Autonomous Agents and Multiagent Systems (AAMAS 2025)}{May 19 -- 23, 2025}
{Detroit, Michigan, USA}{Y.~Vorobeychik, S.~Das, A.~Nowé  (eds.)}
\copyrightyear{2025}
\acmYear{2025}
\acmDOI{}
\acmPrice{}
\acmISBN{}

\usepackage{mathtools}
\usepackage{bm}
\usepackage{tikz}
\usepackage[ruled,linesnumbered,noend]{algorithm2e}
\usepackage{colortbl}
\usepackage{wasysym}
\usepackage{booktabs}
\usepackage{multirow}
\usepackage{nicefrac}  % for nice fractions
\usepackage{enumitem}  % for adjusting lists
\usepackage[capitalise]{cleveref}
\usepackage[toc,page]{appendix}

\usepackage[most]{tcolorbox}

\usetikzlibrary{calc,positioning}

\graphicspath{{figs/}}

\setlist[itemize]{leftmargin=*}

\SetKwInOut{Input}{Input}
\SetKwInOut{Output}{Output}
\SetKwComment{Comment}{\(\triangleright\)\ }{}
\SetCommentSty{itshape}

\Crefname{algorithm}{Alg.}{Algs.}
\Crefname{equation}{Eq.}{Eqs.}
\Crefname{figure}{Fig.}{Figs.}
\Crefname{table}{Tbl.}{Tbls.}
\Crefname{section}{\S\!}{\S\!}
\Crefname{subsection}{\S\!}{\S\!}
\Crefname{subsubection}{\S\!}{\S\!}
\Crefname{appendix}{\S\!\!}{\S\!\!}
\Crefname{theorem}{Thm.}{Thms.}
\Crefname{lemma}{Lem.}{Lems.}
\Crefname{remark}{Rem.}{Rems.}
\Crefname{definition}{Def.}{Defs.}
\Crefname{example}{Ex.}{Exs.}
\Crefname{algorithm}{Alg.}{Algs.}

\DeclareMathOperator{\ReLU}{ReLU}

\DeclareMathOperator{\BP}{\texttt{BP}}

\DeclareMathOperator{\SAT}{\texttt{SAT}_{\texttt{MILP}}}

\newcommand{\NN}{\mathbb{N}}
\newcommand{\RR}{\mathbb{R}}

\newcommand{\QQ}{\mathbb{Q}}
\newcommand{\ZZ}{\mathbb{Z}}

\newcommand{\mDelta}{\bm{\Delta}}

\newcommand{\va}{\bm{a}}

\newcommand{\vc}{\bm{c}}
\newcommand{\vd}{\bm{d}}
\newcommand{\vh}{\bm{h}}
\newcommand{\vo}{\bm{o}}

\newcommand{\vu}{\bm{u}}

\newcommand{\vx}{\bm{x}}
\newcommand{\vy}{\bm{y}}
\newcommand{\vdelta}{\bm{\delta}}
\newcommand{\vomega}{\bm{\omega}}
\newcommand{\vell}{\bm{\ell}}

\newcommand{\sA}{\mathcal{A}}

\newcommand{\sC}{\mathcal{C}}

\newcommand{\sH}{\mathcal{H}}
\newcommand{\sI}{\mathcal{I}}

\newcommand{\sM}{\mathcal{M}}

\newcommand{\sO}{\mathcal{O}}
\newcommand{\sP}{\mathcal{P}}

\newcommand{\sV}{\mathcal{V}}
\newcommand{\sX}{\mathcal{X}}

\newcommand{\tup}[1]{\langle #1 \rangle}
\newcommand{\set}[1]{\{ #1 \}}
\newcommand{\norm}[1]{\lVert #1 \rVert}

\newcommand{\seq}[1]{( #1 )}
\newcommand{\tr}[1]{{#1}^{\intercal}}

\DeclarePairedDelimiter\floor{\lfloor}{\rfloor}

\newtheorem{theorem}{Theorem}
\newtheorem{lemma}{Lemma}

\theoremstyle{example}
\newtheorem{example}{Example}

\theoremstyle{remark}
\newtheorem{remark}{Remark}

\theoremstyle{assumption}
\newtheorem{assumption}{Assumption}

\theoremstyle{definition}
\newtheorem{definition}{Definition}

\newcommand{\celluns}{\cellcolor{red!40}}
\newcommand{\cellsat}{\cellcolor{blue!30!green!30}}

\newcommand{\Next}{\bigcirc}
\newcommand{\Until}{\,\mathcal{U}}
\newcommand{\Release}{\mathcal{R}}
\newcommand{\Always}{\square}
\newcommand{\Finally}{\lozenge}

\newcommand{\True}{\texttt{True}}
\newcommand{\False}{\texttt{False}}
\newcommand{\Unknown}{\texttt{Unknown}}
\newcommand{\LTL}{\(\text{LTL}_{\RR}\)}

\newcommand{\xmark}{\({\color{red!70!black}\mathsf{x}}\)}
\newcommand{\cmark}{{\color{green!80!blue!70!black}\checkmark}}

\newcommand{\sysname}{{MN-MAS}}
\newcommand{\fullsysname}{Memoryful Neural Multi-Agent System}

\definecolor{rose}{rgb}{1,0.33,0.44}

\acmSubmissionID{<<977>>}

\title[LTL Verification of Memoryful Neural Agents]{LTL Verification of Memoryful Neural Agents}

\author{Mehran Hosseini}
\affiliation{
  \institution{King's College London}
  \city{London}
  \country{United Kingdom}}
\email{mehran.hosseini@kcl.ac.uk}

\author{Alessio Lomuscio}
\affiliation{
  \institution{Imperial College London}
  \city{London}
  \country{United Kingdom}}
\email{a.lomuscio@imperial.ac.uk}

\author{Nicola Paoletti}
\affiliation{
  \institution{King's College London}
  \city{London}
  \country{United Kingdom}}
\email{nicola.paoletti@kcl.ac.uk}

\begin{abstract}
  We present a framework for verifying Memoryful Neural Multi-Agent Systems (MN-MAS) against full Linear Temporal Logic (LTL) specifications. In MN-MAS, agents interact with a non-deterministic, partially observable environment. Examples of MN-MAS include multi-agent systems based on feed-forward and recurrent neural networks or state-space models. Different from previous approaches, we support the verification of both bounded and unbounded LTL specifications. We leverage well-established bounded model checking techniques, including lasso search and invariant synthesis, to reduce the verification problem to that of constraint solving. To solve these constraints, we develop efficient methods based on bound propagation, mixed-integer linear programming, and adaptive splitting. We evaluate the effectiveness of our algorithms in single and multi-agent environments from the Gymnasium and PettingZoo libraries, verifying unbounded specifications for the first time and improving the verification time for bounded specifications by an order of magnitude compared to the SoA.
\end{abstract}


\newcommand{\BibTeX}{\rm B\kern-.05em{\sc i\kern-.025em b}\kern-.08em\TeX}

\begin{document}

%%% The following commands remove the headers in your paper. For final 
%%% papers, these will be inserted during the pagination process.

\pagestyle{fancy}
\fancyhead{}

%%% The next command prints the information defined in the preamble.

\maketitle 

%%%%%%%%%%%%%%%%%%%%%%%%%%%%%%%%%%%%%%%%%%%%%%%%%%%%%%%%%%%%%%%%%%%%%%%%

%
%%
\section{Introduction}
\label{sec: Introduction}
With the rise of machine learning, we are increasingly witnessing applications of neural networks (NNs) to autonomous systems. While \emph{neural agents} may solve significantly more complex tasks than traditionally programmed
counterparts~\cite{Silver+17,Vinyals+19}, NNs suffer from fragilities~\cite{Szegedy+14}, which is a concern in safety-critical applications. Thus, they require rigorous guarantees, as those offered by formal verification.

In this paper, we develop the methodology for verifying systems comprising multiple memoryful neural
agents that interact with an uncertain environment over time in a
closed loop. We refer to such systems as \emph{\fullsysname}
(\emph{\sysname}).

The verification of \sysname{} is more challenging than the
verification of individual predictions made by NNs, due to the
presence of multiple networks, multiple time steps, and temporal
dependencies between predictions. Several methods for the verification
of systems with NNs ``in the loop'' have been recently
proposed~\cite{HosseiniL23,FanHCL020}.  Nevertheless, these are limited by at least two of the following: they only consider single-agent systems; are limited to the analysis of (time-)bounded reachability; do not
account for uncertainty in the environment; or the agents are
implemented by feedforward models, thus are stateless/memoryless.

We propose a verification approach for \sysname{} that addresses all the
above limitations. The approach presented here supports
% AL: here and everywhere: we do not do model checking here. We simply
% do verification.
verification of \sysname{} against full \emph{linear
temporal logic}~\cite{Pnueli77}. Thus, the method presented here goes beyond simple
bounded reachability analysis and, unlike existing methods, can verify both bounded and \textit{unbounded} specification, as well as nested properties. These include (unbounded) safety (e.g., \textit{will an agent always remain in a safe state region?}) or stability (e.g., \textit{will an agent reach its target and remain there long enough thereafter?}).

Although we show that the verification of unbounded specifications is
undecidable in general, we develop sound or complete procedures based
on Bounded Model Checking (BMC)~\cite{Biere+09}, which allows us to
verify unbounded properties of the \sysname{} using finite paths. In
particular, we develop three BMC procedures, based on, respectively,
simple path unrolling, lasso search, and inductive invariants.
% by verifying \sysname{} on finite traces, which is further enhanced
% by \emph{lasso} search and \emph{inductive invariant} synthesis.
Thus, when the presented method establishes that an \sysname{} $S$
satisfies an LTL specification $\phi$, written $S \models \phi$, then
$\phi$ must hold for all the paths of the system induced by the
uncertain (non-deterministic) environment dynamics and initial states.

With BMC, we reduce the verification problem to solving a set of
logical constraints over the \sysname{} states.  In particular, we
target linearly definable systems, i.e., systems that can be
formulated by
a finite set of piecewise-linear (PWL) constraints.  This class of
systems remains quite expressive, as it includes, among others, neural
agents with ReLU-based activations.

We introduce two efficient algorithms to solve the resulting PWL
constraints. The first is \emph{Bound Propagation through
  Time (BPT)}, an extension of interval and linear bound
propagation~\cite{Gowal+18,Wang+19} to support memoryful models, such
as those based on \emph{Recurrent Neural Networks} (\emph{RNNs}). BPT
is highly efficient but results in an over-approximation of the
model's outputs, meaning that it cannot be used in all verification
instances (see \cref{subsec: Bounded Solving}). The second is a
\emph{Mixed-Integer Linear Programming} (\emph{MILP})-based algorithm,
called \emph{Recursive MILP} (\emph{RMILP}), which allows encoding the
constraints and variables during multiple time steps in MILP. Unlike
BPT, the latter is precise (i.e., does not overapproximate) and is
more efficient than past MILP approaches for recurrent models
\cite{HosseiniL23,Akintunde+19} (see \cref{subsec: Lunar Lander}).
We also employ \emph{adaptive splitting} \cite{KouvarosL21},
which exploits the information about the reachable sets computed by
BPT to resolve \textit{a priori} some of the discrete variables of the
MILP instances, thereby considerably simplifying the MILP encoding.

In summary, we make the following contributions: \textbf{(1)} we introduce the first framework for verifying memoryful,
  neural multi-agent systems against full LTL specifications under
  uncertainty;
\textbf{(2)} we characterize the decidability of the verification problem;
\textbf{(3)} we develop efficient solutions based on bounded model checking,
  bound propagation, and constraint solving; and
\textbf{(4)} we evaluate the approach on deep Reinforcement Learning (RL)
  benchmarks with single and multiple agents, demonstrating its
  efficiency compared to the SoA (being
    strictly faster in all instances, up to one order of magnitude)
  and its versatility in verifying, for the first time, both time-bounded
  and unbounded specifications.

{\large\textbf{Related Work.}}
Following several seminal works exposing the
vulnerability of machine learning models to adversarial
perturbations~\cite{Szegedy+14,BiggioR18}, significant attention has
been paid to the verification of neural networks against input-output
relations. A special class of these consists of the verification of
local robustness including output reachability.
Most solution techniques are based on (a combination of) bound
propagation~\cite{Gowal+18,Wang+19} and
constraint solving
\cite{Ehlers17,LomuscioM17,TjengXT19,Katz+19,Botoeva+20}.

While this line of research focuses on open-loop systems, we here consider
closed-loop systems with NN components, i.e., sequential processes
characterised by multiple interdependent NN predictions over time. In
recent years, several methods have been proposed to verify this kind
of systems
\cite{HosseiniL23,AkintundeBKL20-AAMAS,AkintundeBKL20-KR,Tran+20,Sidrane+22,FanHCL020,BacciP20,Ivanov+21,Wicker+21,Wicker+23,Giacomarra+2025-Generative}. 

Among the above, the work of \citet{AkintundeBKL20-KR}
focuses on temporal logic verification of multi-agent systems but
considers a bounded fragment of the logic and agents implemented via
feedforward NNs (FFNN), i.e., memoryless models. Closer to this
contribution is the method of \citet{Akintunde+19}, which supports RNN-based agents
but is limited to single-agent systems and bounded LTL properties, and
its recent extension by \citet{HosseiniL23}, which introduces an inductive
invariant algorithm for unbounded safety. 
Our work considerably
expands on these previous proposals in that it is the first to support multi-agent systems comprising neural memoryful
agents and partially-observable uncertain environments, as well as arbitrary LTL
specifications, including unbounded and nested formulas.

Verifying recurrent models on the other hand is significantly more
challenging---even for recurrent linear models w.r.t.
time-unbounded linear specifications \cite{Hosseini2021-Thesis,HosseiniOW2019-Universal,Almagor+2018-Divergence,OuaknineW2014-Problems,Braverman2006-Rational,Tiwari2004-Real}.

To verify models with recurrent layers, one approach is to explicitly
unroll the RNN to derive an equivalent
FFNN~\cite{Akintunde+19}. However, the size of the resulting network
blows up with the unrolling depth (i.e., the time bound). This problem
is mitigated in~\cite{HosseiniL23} by introducing a (more efficient)
direct MILP encoding of the RNN constraints, but still results in
often unfeasible queries. Other related RNN verification methods, but
applied to open-loop systems, include~\cite{Ko+19}, based on linear
bound propagation, and~\cite{JacobyBK20}, which derives a
time-invariant overapproximation of the RNN.  Our solution extends the MILP encoding of \cite{HosseiniL23} to allow for multiple agents and composite \LTL{} formulae, and further enhances it with bound
propagation techniques and adaptive splitting \cite{KouvarosL21,HenriksenL20}, as we discuss in \cref{subsec: Bounded Solving}.

A related line of work concerns the verification of multi-agent
reinforcement learning (MARL) systems with neural agents \cite{MqirmiBL21,Riley+21,Yan+22-synthesis,Yan+22-equilibria}. The main
difference is that in MARL the environment is stochastic while in
\sysname it is non-deterministic.

%%% Local Variables:
%%% mode: latex
%%% TeX-master: "../main"
%%% TeX-master: "../main"
%%% End:

%%% Local Variables:
%%% mode: latex
%%% TeX-master: "../main"
%%% End:

%
%%
\section{Background \& Problem Formulation}
\label{sec: Background}
We denote the \(n\)-dimensional vector space over \(\RR\) by
\(\RR^n\). Vectors \(\vx \in \RR^n\) are shown by bold italic typeface
while scalars \(x \in \RR\) are shown by normal italic typeface. We
denote the set of all finite sequences over a set
\(\sC \subseteq \RR^n\) by \(\sC^*\) and the set of all countable
infinite sequences over \(\sC\) by \(\sC^{\omega}\).  We use
\(\seq{\vx^{\tup{i}}}_{i = 0}^n\) and
\(\seq{\vx^{\tup{i}}}_{i \in \NN}\) to indicate a sequence in
\((\RR^n)^*\) and \((\RR^n)^{\omega}\), respectively. The \(t\)-th
element of such sequence denoted by \(\vx^{\tup{t}}\). For a function
\(f\), its \(t\) consecutive applications is denoted by
\(f^{(t)}\). For the sake of brevity, we treat sets like vectors and
scalars; e.g., for a set \(\sP \subseteq \RR^n\), we use
\(\tr{\vc} \cdot \sP\) to refer to the set
\(\set{\tr{\vc} \cdot \vx : \vx \in \sP}\).

Next, we define PWL
constraints, and then, in the rest of the section, we lay out the
problem formulation.
\begin{definition}
  \label{def: PWL}
  A \emph{Piecewise-Linear} (\emph{PWL}) \emph{constraint} is a
  constraint of the form
  \begin{equation}
  \label{eq: PWL}
    A_{\vx} \vx + A_{\vy} \vy + A_{\vdelta} \vdelta \leq \vd,
  \end{equation}
  where \(\vx \in \RR^{n_{\vx}}, \vy \in \RR^{n_{\vy}}\), and
  \(\vdelta \in \mDelta \subseteq \ZZ^{n_{\vdelta}}\) are vector
  variables (or simply variables), and
  \(A_{\vx} \in \QQ^{n_c \times n_{\vx}}, A_{\vy} \in \QQ^{n_c \times
    n_{\vy}}\), \(A_{\vdelta} \in \QQ^{n_c \times n_{\vdelta}}\), and
  $\vd \in \QQ^{n_c}$ are real-valued matrices. The variable
  \(\vdelta\) is called the \emph{discrete variable} and often
  \(\mDelta = \set{0, 1}^{n_{\vdelta}}\), i.e., \(\vdelta\) is a
  Boolean variable. Variables \(\vx\) and \(\vy\) are the \emph{free}
  and \emph{pivot} variables, respectively. In other words,
  \cref{eq: PWL} defines \(\vy\) linearly in terms of
  \(\vx\), in a piecewise fashion because of the discrete variable
  \(\vdelta\). We say that \(\vy\) as a PWL functions of \(\vx\) if
  for each \(\vy \in \RR^{n_{\vy}}\), there exists at most one
  \(\vx \in \RR^{n_{\vx}}\) that satisfies \cref{eq: PWL}.
\end{definition}
\subsection*{Problem Formulation}
We consider systems comprising a non-deterministic and partially
observable environment with multiple agents, controlled by
arbitrary PWL neural policies. Such policies can be memoryful,
e.g., RNN-based, or without memory, e.g., FFNN-based. In this paper,
we focus on the verification of specifications expressed in
\emph{Linear Temporal Logic} over the \emph{reals}, \LTL, an extension
of LTL \cite{Pnueli77} with atomic predicates involving real
variables, as defined in \cref{def: Specifications}.

\begin{definition}[\fullsysname]
  \label{def: System}
  A \emph{\fullsysname} (\emph{\sysname}) \(S\) comprises \(s\)
  \emph{agents} interacting over time with a non-deterministic
  \emph{environment}, as per following equations.
  \begin{align}
    \vx^{\tup{t}} \in & \ \tau(\vx^{\tup{t-1}}, \va_1^{\tup{t-1}},\ldots, \va_s^{\tup{t-1}}),
                        \label{eq:aes-1}\\
    \va_i^{\tup{t}} = & \ \alpha_i(\vo_i^{\tup{1}}, \ldots, \vo_i^{\tup{t}}),
                        \label{eq:aes-2}\\
    \vo_i^{\tup{t}} = & \ o_i(\vx^{\tup{t}}).
                        \label{eq:aes-3}
  \end{align}
  In \cref{eq:aes-1,eq:aes-2,eq:aes-3},
  \begin{itemize}
  \item $\vx^{\tup{t}} \in \sX$ is the \emph{state} of the environment
    at time $t$ and $\vx^{\tup{0}} \in \sX^{\tup{0}}$ is the
    \emph{initial state} of the environment, with $\sX\subseteq \RR^m$
    being the environment's \emph{state space} and
    $\sX^{\tup{0}}\subseteq \sX$ the set of \emph{initial states};
  \item $(\va_1^{\tup{t}},\ldots, \va_s^{\tup{t}}) \in \sA$ are the
    \emph{actions} performed at time $t$ by agents $1,\ldots, s$
    respectively, where $\sA=\sA_1 \times \dots \times \sA_s$ and
    $\sA_i$ is the \emph{action space} of agent $i$; and
  \item \(\tau: \sX \times \sA \rightarrow 2^{\sX}\) is the
    environment's \emph{transition relation} which, given the current
    state of the environment \(\vx \in \sX\) and agents' actions
    \(\va \in \sA\), returns the set of possible next states
    \(\tau(\vx, \va) \subseteq \sX\).
  \end{itemize}
  Moreover, for each agent $i=1,\ldots,s$,
  \begin{itemize}
  \item $\vo_i^{\tup{t}} \in \sO_i$ is agent $i$'s \emph{observation}
    of the environment at $t$, with $\sO_i \subseteq \RR^\ell$ being the
    agent's \emph{observation space}; observations are made by an
    \emph{observation function} \(o_i: \sX \rightarrow \sO_i\);
  \item \(\alpha_i: \sO_i^* \rightarrow \sA_i\) is the agent \(i\)'s
    \emph{policy}, which given a finite sequence of environment
    observations returns an admissible action from \(\sA_i\). The
    policy $\alpha_i$ is realised via a \emph{recurrent function}
    $r_i:\sO_i \times \sH_i \to \sA_i \times \sH_i$, where
    \(\sH_i \subseteq \RR^{n_i}\) is the agent's \emph{set of hidden
      states}. The function $r_i$, defined
    \begin{equation}
      \label{eq: Memoryful Policy}
      r_i(\vo_i^{\tup{t}}, \vh_i^{\tup{t}}) = (\va_i^{\tup{t}}, \vh_i^{\tup{t+1}}),
    \end{equation}
    maps observation \(\vo_i^{\tup{t}}\) and agent's hidden state
    \(\vh_i^{\tup{t}}\) to action \(\va_i^{\tup{t}}\) and updated
    state \(\vh_i^{\tup{t+1}}\).  When \(t=0\),
    \(\vh_i^{\tup{0}}\in \sH_i^{\tup{0}}\), the set of \emph{initial
      hidden states} of agent \(i\). We refer to $\vh_i$ as the
    \emph{memory} or \emph{hidden state} of the agent.
  \end{itemize}
  When the transition relation \(\tau\) is a function 
  \(\tau: \sX \times \sA \rightarrow \sX\), the environment is \emph{deterministic}. We say that agent \(i\) is
  \emph{memoryless} or \emph{feed-forward} whenever
  \(\sH_i = \emptyset\). As done for \(\sA\), we define
  \(\sO = \sO_1 \times \dots \times \sO_s\) and
  \(\sH = \sH_1 \times \dots \times \sH_s\) to be the \emph{sets of
    systems's joint observations} and \emph{hidden states},
  respectively. We denote by \(\vh \in \sH\) and \(\vo \in \sO\) a
  \emph{joint observation} and a \emph{joint hidden state},
  respectively. Consequently, we denote the \emph{initial set of joint
    agent states} by \(\sH^{\tup{0}}\). This allows us to define the
  \emph{joint observation function}
  \(o(\vx) = (o_1(\vx), \dots, o_s(\vx))\) and \emph{joint policy}
  \(r(\vo, \vh) = (\va, \vh') = (r_1(\vo_1, \vh_1), \dots, r_s(\vo_s,
  \vh_s))\). We can now define the \emph{system evolution}
  \(e(\vx, \vh) = (\tau(\vx, \va), \vh')\), which given the current joint system state \(\vx\) and joint memory \(\vh\), returns the next joint system state \(\tau(\vx, \va)\) and joint memory \(\vh'\).
\end{definition}
\begin{assumption}
  We will focus on \sysname{} that are linearly definable or can be
  linearly approximated, that is, $\tau$, $r$,
  $o$, and the sets of initial states
  $\sX^{\tup{0}}, \sH^{\tup{0}}$, can be
  expressed by finite conjunction and disjunction of linear
  inequalities. This assumption is insignificant as it
  captures NN policies with ReLU-like activations, among others. Moreover, non-linear
  systems can be linearly approximated to an arbitrary level of
  precision as described in \cite{AkitundeLMP18,DAmbrosioLM10}.
\end{assumption}
\begin{remark}
  In the experiments, we implement memoryful policies in \cref{eq:
    Memoryful Policy} using RNNs and memoryless policies using FFNNs. Our
  formulation however is more general and can accommodate other kinds
  of memoryful models, such as LSTM \cite{HochreiterS97}, memory
  networks \cite{WestonCB14}, and state space models
  \cite{GhahramaniH00}, such as Mamba \cite{GuD23}.
\end{remark}
\begin{definition}[\sysname\ path]
  A sequence $\pi =(\vx^{\tup{0}},\bm{\vh}^{\tup{0}},\vx^{\tup{1}},$
  $\bm{\vh}^{\tup{1}},\ldots) \in (\sX\times \sH)^{\omega}$ is a path
  of an \sysname\ $S$ if, for any $t\geq 0$, $\vx^{\tup{t+1}}$ and
  $\bm{\vh}^{\tup{t+1}}$ are recursively derived from $\vx^{\tup{t}}$
  and $\bm{\vh}^{\tup{t}}$ following
  \cref{eq:aes-1,eq:aes-2,eq:aes-3,eq: Memoryful Policy}. We denote the
  set of all infinite paths of $S$ by
  \(\Pi_S \subseteq (\sX\times \sH)^{\omega}\). For a path
  $\pi \in \Pi_S$ and index $t$, we denote with
  $\pi_{\sX}^{\tup{t}} \in \sX$ and $\pi_{\sH}^{\tup{t}} \in \sH$ the
  environment state and agents' hidden states, respectively, at time
  $t$ in $\pi$. For a subset of systems paths \(\sP \subseteq \Pi_S\),
  we denote its projection onto \(\sX\) by \(\sP_{\sX}\), defined as
  \(\sP_{\sX} = \set{\vx \in \sX: (\vx, \vh) \in \sP \text{ for some } \vh \in \sH}\). The projection of \(\sP\) onto \(\sH\) is denote by \(\sP_{\sH}\) and is defined similarly.
\end{definition}

We can now proceed to introduce the \LTL{} specifications.
\begin{definition}[Specifications]
  \label{def: Specifications}
  The \emph{Linear Temporal Logic} over the \emph{reals}, \LTL{}, is
  defined by following BNF.
  \begin{equation*}
  \resizebox{\linewidth}{!}{$
    \displaystyle
    \phi ::= \tr{\vc}\! \cdot \vx \leq d \,|\, \neg \phi \,|\,
    \phi \wedge \phi \,|\, \phi \vee \phi \,|\,
    \Next \phi \,|\, \phi \Until \phi \,|\,
    \phi \Release \phi \,|\, \Always \phi \,|\, \Finally \phi,
    $}
  \end{equation*}
  where \(\vc \in \RR^m\) and \(d \in \RR\) are constants. \LTL{} atoms are linear inequalities of the form
  $\tr{\vc} \cdot \vx \leq d$ involving the environment state
  $\vx$.  If a formula
  contains any of the temporal operators
  \(\Until, \Release, \Always, \Finally\), it is an
  \emph{unbounded formula}; otherwise, a \emph{bounded formula}.
\end{definition}
\begin{remark}
  \LTL{} atoms are defined over the environment's states as we are not
  interested in reasoning about the hidden states of agents, although these can be trivially included in the logic if required.
\end{remark}

The formula \(\Next \phi\) is read as ``\(\phi\) holds at the next
time step''; \(\psi \Until \phi\) is read as ``\(\psi\) holds at least
until \(\phi\) becomes true, which must hold at the current or a
future position''; \(\psi \Release \phi\) is read as ``\(\phi\) has to
be true until and including the point where \(\psi\) first becomes
true; if \(\psi\) never becomes true, \(\phi\) must remain true
forever''; \(\Always \phi\) is read as ``\(\phi\) always holds''; and
\(\Finally \sC\) is read as ``\(\sC\) holds at some point in the
future''.

For the sake of brevity, we use \(\Next^k \phi\), to indicate \(k\)
consecutive applications of \(\Next\), i.e., to state ``\(\phi\) holds
after \(k\) steps''. For an unbounded operator, such as \(\Until\), we denote with $\Until^{\leq k}$ its bounded version; e.g., \(\psi \Until^{\leq k} \phi\) reads ``\(\phi\)
holds within \(k\) time steps and \(\psi\) holds up to then.'' Note
that \(\psi \Until^{\leq k} \phi\) is, in fact, a bounded formula,
which can be expressed using Boolean operators and finitely many
\(\Next\); specifically, \(\psi \Until^{\leq k} \phi = \bigvee_{i=0}^k ((\bigwedge_{j=0}^{i-1}\Next^{j} \phi) \wedge \Next^{i} \psi)\).

Next, the \LTL{} satisfaction relation is defined as follows.

\begin{definition}[Satisfaction]
  \label{def: Satisfaction}
  The satisfaction relation \(\models\) for a path \(\pi \in \Pi_S\)
  of an \sysname\ \(S\), time step $t$, and \LTL{} formulae \(\phi\)
  and \(\psi\) is defined as follows.
  \begin{equation*}
    \arraycolsep=2pt
    \begin{array}{rlcl}
      (\pi,t) \models & \tr{\vc} \cdot \vx \leq d  & \text{iff} & \tr{\vc} \cdot \pi_{\sX}^{\tup{t}} \leq d;\\[2pt]
      (\pi,t) \models & \psi \vee \phi  & \text{iff} &  (\pi,t) \models \psi \text{ or } (\pi,t) \models \phi;\\[2pt]
      (\pi,t) \models & \psi \wedge \phi  & \text{iff} &  (\pi,t) \models \psi \text{ and } (\pi,t) \models \phi;\\[2pt]
      (\pi,t) \models & \lnot \phi  & \text{iff} &  (\pi,t) \not\models \phi;\\[2pt]
      (\pi,t) \models & \Next^k \phi  & \text{iff} & (\pi, t+k) \models \phi;\\[2pt]
      (\pi,t) \models & \Always \phi  & \text{iff} &  \forall i \geq 0, \ (\pi, t+i) \models \phi;\\[2pt]
      (\pi,t) \models & \Finally \phi  & \text{iff} &  \exists i \geq 0, \ (\pi,t+i) \models \phi;\\[2pt]
      (\pi,t) \models & \psi \Until \phi  & \text{iff} &  \exists i \geq 0, \ (\pi,t+i) \models \phi \  \wedge \\
                      &&& \bigwedge_{j=0}^{i-1} (\pi,t+j) \!\models \psi;\\[4pt]
      (\pi,t) \models & \psi \Release \phi & \text{iff} & \forall i \geq 0, \ (\pi,t+i) \!\models \phi; \vee \\
      \multicolumn{4}{r}{\exists i \geq 0, \ (\pi,t+i) \models \psi \wedge \bigwedge_{j=0}^{i} (\pi,t+j) \!\models \phi. }
    \end{array}
  \end{equation*}
  We say that a set of environment states $\sX'\subseteq \sX$ and a
  set of hidden states $\sH' \subseteq \sH$ satisfy $\phi$, written
  $\sX' \times \sH' \models \phi$, if $(\pi,0)\models \phi$ for all
  paths $\pi \in \Pi_S$ where $\pi_{\sX}^{\tup{0}}\in \sX'$ and
  $\pi_{\sH}^{\tup{0}}\in \sH'$ (i.e., for all paths whose initial
  states are in $\sX'$ and $\sH'$). We say that an \sysname{} $S$
  satisfies $\phi$, written \(S \models \phi\), if
  $\sX^{\tup{0}} \times \sH^{\tup{0}} \models \phi$, i.e., if every
  \sysname\ path that starts from an initial state of $S$ satisfies
  \(\phi\).
\end{definition}

We can now formulate the \LTL{} \emph{verification problem} for \sysname,
\emph{the problem we address in this paper}.

%
% \begin{tcolorbox}[colback=gray!15,colframe=black,arc=5pt,boxrule=1pt]
\begin{definition}[Verification Problem]
  \label{def: Model Checking}
  Given an \sysname\ \(S\) and formula \(\phi \in\) \LTL, determine if
  \(S \models \phi\).
\end{definition}
% \end{tcolorbox}
%
\noindent
When \(\phi\) is bounded, we call the problem above the \emph{Bounded Verification Problem} (\emph{BVP}), and when \(\phi\) is unbounded, we call it the \emph{Unbounded Verification Problem} (UVP). If an instance \(S \models \phi\) can be solved, we say that it is \emph{solvable}.

As we discussed in the related work, the problem considered in \cref{def: Model Checking} is more general than previous works as it allows multiple agents, arbitrary \LTL{} formulae, and non-deterministic partially-observable environments.
\vspace{-0.3em}

%%% Local Variables:
%%% mode: latex
%%% TeX-master: "../main"
%%% End:

% 
%%
\section{\sysname{'} Bounded \LTL Verification}
\label{sec: Bounded Verification}
Here, we introduce our approach for solving the bounded
verification problem for \sysname{}. Particularly, In \cref{subsec: 
Bounded Complexity}, we show that BVP is decidable.
In \cref{subsec: Bounded Solving}, we provide two algorithms
for solving BVP.

\subsection{Decidability of Bounded \LTL{} Verification}
\label{subsec: Bounded Complexity}
We start by proving the decidability of verifying
atomic \LTL{} propositions and $\Next$ formulas in \cref{lem: Solving
  Atomic,lem: Solving Next}.
\begin{lemma}
  \label{lem: Solving Atomic}
  Verifying \(S \models \phi = (\tr{\vc} \cdot \vx \leq d)\) is decidable.
\end{lemma}
\begin{proof}
  From \cref{def: Satisfaction,def: Model Checking} we know that
  \(S \models \phi\) iff \(\tr{\vc} \cdot \sX^{\tup{0}} \leq d\),
  where \(\sX^{\tup{0}}\) is the set of initial states of the
  environment. In other words, we want to check whether
  \(\forall \vx \in \sX^{\tup{0}} \ (\tr{\vc} \cdot \vx \leq d)\),
  which can be solved using MILP \cite{BulutR21}
  because it entails checking the unsatisfiability of
  \(\exists \vx \in \sX^{\tup{0}} \ \neg(\tr{\vc} \cdot \vx \leq d)\).
\end{proof}
\begin{lemma}
  \label{lem: Solving Next}
  Verifying \(S \models \phi = \Next^t (\tr{\vc} \cdot \vx \leq d)\) is
  decidable.
\end{lemma}
\begin{proof}
  As we showed in the proof of \cref{lem: Solving Next}, solving
  \(S \models \phi\) for \(t = 0\) amounts to
  verifying
  \(\forall \vx \in \sX^{\tup{0}} \ (\tr{\vc} \cdot \vx \leq d)\).
  For \(t \geq 1\), it amounts to verifying
  \begin{equation}
    \label{eq: Next t}
    \forall \vx \in \sX^{\tup{t}} \ (\tr{\vc} \cdot \vx \leq d),
  \end{equation}
  where \(\sX^{\tup{t}}\) is recursively defined via
  \begin{equation}
    \label{eq: Recursive X}
    \sX^{\tup{t}} = \sP_{\sX}^{\tup{t}}, \quad
    \sP^{\tup{t}} = \set{e(\vx, \vh): (\vx, \vh) \in \sP^{\tup{t-1}}}.
  \end{equation}
  In \cref{eq: Recursive X}, the number of free, pivot, and discrete
  variables defining \(\sX^{\tup{t-1}}\) grows linearly with
  \(t\). Specifically, if the defining constraints of the \sysname{}
  have \(m_{\vx}, m_{\vy}\), and \(m_{\vdelta}\) free, pivot, and
  discrete variables, respectively, then the defining constraints of
  \(\sX^{\tup{t}}\) have at most \((t+1)m_{\vx}, (t+1)m_{\vy}\), and
  \((t+1)m_{\vdelta}\) free, pivot, and discrete variables,
  respectively. Therefore verifying \cref{eq: Next t} is decidable,
  which implies verifying
  \(S \models \Next^t \ (\tr{\vc} \cdot \vx \leq d)\) is decidable.
\end{proof}

This readily gives us the following result on the decidability of
verifying Boolean combinations of \LTL{} formulae.
\begin{lemma}
  \label{lem: Solving Boolean}
  Given \LTL{} formulae \(\phi\) and \(\psi\), such that
  \(S \models \phi\) and \(S \models \psi\) are decidable,
  \(S \models \phi \wedge \psi, S \models \phi \vee \psi\)
  and
  \(S \models \lnot\phi\)
  are also all decidable.
\end{lemma}

We can now prove that BVP is decidable.
%
% \begin{tcolorbox}[colback=gray!15,colframe=black,arc=5pt,boxrule=1pt]
\begin{theorem}
  \label{thm: Bounded}
  Bounded \LTL{} verification of \sysname{} is decidable.
\end{theorem}
% \end{tcolorbox}
%
\begin{proof}
  Bounded \LTL{} formulae consist of atomic propositions alongside
  \(\neg, \wedge, \vee\), and \(\Next\) operators. Using the duality
  of \(\Next\) and \(\neg\) and the distributivity
  of \(\Next\) on \(\wedge\) and \(\vee\),\footnote{Duality means
    \(\neg \! \Next \! \phi \! \equiv \! \Next \neg \phi\), and
    distributivity means
    \(\Next (\phi \wedge \psi) \equiv \Next \phi \wedge \Next \psi\)
    and \(\Next (\phi \vee \psi) \equiv \Next \phi \vee \Next \psi\)
    for all \(\phi, \psi \in\) \LTL.}  we can rewrite every bounded
  \(\phi \in\) \LTL{} formula as a finite (linear in the length of
  \(\phi\)) combination (using \(\wedge, \vee\), and \(\neg\)) of
  formulae of the form \(\Next^{i} (\tr{\vc} \cdot \vx \leq d)\). It
  follows from \cref{lem: Solving Atomic,lem: Solving Next,lem:
    Solving Boolean} that verifying \(S \models \phi\) is
  decidable.
\end{proof}
\subsection{Solving Bounded \LTL{} Verification}
\label{subsec: Bounded Solving}
Above, we outlined the procedure for solving BVP in \cref{subsec:
  Bounded Complexity} at a high level. We now provide two methods for
solving the BVP algorithmically. The first is \emph{Bound Propagation
  through Time} (\emph{BPT}),
% which is an extension of bound propagation \cite{Gowal+18,Wang+19} allowing to propagate the bounds not just forward, through networks' layers, but also through time. BPT 
which is a sound but incomplete (i.e., overapproximated) technique for solving BVP. Next, we introduce \emph{Recursive MILP} (\emph{RMILP}), which allows encoding the BVP as an MILP problem, thereby providing a sound and complete (i.e., exact) solution to BVP.

\begin{algorithm}[t]
  \caption{Bound Propagation through Time}
  \label{alg: BPT}
  \DontPrintSemicolon

  \Input{\sysname{} \(S\), \(\phi=\Next^i (\tr{\vc} \cdot \vx \leq d) \in \)
    \LTL} \Output{\texttt{True/Unknown}}

  \(\vell^{\tup{0}}, \vu^{\tup{0}} = \min_{\sX^{\tup{0}} \times \sH^{\tup{0}}}(\vx), \max_{\sX^{\tup{0}} \times \sH^{\tup{0}}}(\vx)\)

  \For{\(t \gets 1, \dots, i\)}{

    \(\vell^{\tup{t}}, \vu^{\tup{t}} = \BP(e(\cdot),
    (\vell^{\tup{t-1}}, \vu^{\tup{t-1}})) \)

  }
  \If{\((\tr{\vc} \cdot \vell_{\sX}^{\tup{i}} \leq d) \wedge (\tr{\vc} \cdot \vu_{\sX}^{\tup{i}} \leq d)\)}{
    \Return \True
  }
  \Return \Unknown
\end{algorithm}
\subsubsection{Bound Propagation through Time (BPT)}
\label{subsec: BPT}
\emph{Bound propagation} \cite{Gowal+18,Wang+19} is commonly used for verifying the robustness
of neural networks. To use it for memoryful models,
such as RNNs, as well as systems with temporal dependencies, such as
\sysname, a more refined approach is required. We introduce
\emph{BPT}, an extension of bound propagation which allows us to propagate the network, memory, and
system bounds through time.

The high-level procedure behind BPT is outlined in \cref{alg:
  BPT}. First, we note that every bounded formulae in \LTL{} is a
finite composition (using \(\wedge, \vee, \neg\)) of formulae of the
form \(\Next^i (\tr{\vc} \cdot \vx \leq d)\), and so, w.l.o.g.\ we
illustrate BPT for checking the latter formulae. \cref{alg:
  BPT} starts with the bounds \(\vell^{\tup{0}}, \vu^{\tup{0}}\) of
the joint initial system and hidden state. It then uses a bound
propagation method \(\BP\) (such as interval or linear bound
propagation) to propagate \(\vell^{\tup{1}}, \vu^{\tup{1}}\) through
the combined system evolution function \(e(\cdot)\) to compute bounds
\(\vell^{\tup{1}}, \vu^{\tup{1}}\) on the joint system and hidden
state at the next time step (line 3).  This is repeated iteratively to
obtain bounds for any finite time step. Finally, the algorithm checks
whether the final bounds \(\vell^{\tup{i}}, \vu^{\tup{i}}\) satisfy
the given constraints (line 4).

Using bound propagation we have that, at each step $t$, \([\vell^{\tup{t}}, \vu^{\tup{t}}]\) \(\supseteq e([\vell^{\tup{t-1}}, \vu^{\tup{t-1}}] )\), where $[\vell^{\tup{t}}, \vu^{\tup{t}}]$ is the hyper-box defined by the bounds $\vell^{\tup{t}}$ and $\vu^{\tup{t}}$. In other words, BPT computes an over-approximation of the reachable states; thus, the algorithm returns \Unknown\ if the BPT bounds do not satisfy the constraints. 

Let us demonstrate BPT using an
example. We use an RNN for the sake of simplicity
rather than the full system evolution function.

\begin{algorithm}[t]
  \caption{Recursive MILP}
  \label{alg: RMILP}
  \DontPrintSemicolon
  \SetKwFunction{proc}{ buildConstraints}
  \Input{\sysname{} \(S\), \(\phi=\Next^i (\tr{\vc} \cdot \vx \leq d) \in \)
    \LTL}
  \Output{\texttt{True/False}}
    $(\sV, \sC)=$\proc{$i$}\;
    \Return \(\SAT(\forall \vx^{\tup{i}} . \   \tr{\vc} \cdot \vx^{\tup{i}} \leq d, \sV, \sC)\)\;
    \;

    \SetKwProg{myproc}{Procedure}{}{}
  \myproc{\proc{$t$}}{
  
  $\sV'=\{\vx^{\tup{t}}, \vh^{\tup{t}}\}$\;
  \If{$t=0$}{
    \Return $(\sV',\{\vx^{\tup{0}}\in \sX^{\tup{0}}, \vh^{\tup{0}}\in \sH^{\tup{0}}  \})$
  }
  \Else{
    $(\sV, \sC)=$\proc{$t-1$}\;
    $\sC' = \{(\vx^{\tup{t}}, \vh^{\tup{t}}) = e(\vx^{\tup{t-1}}, \vh^{\tup{t-1}})\}$\;
    \Return $(\sV\cup \sV', \sC \cup \sC')$
  }
  }  
\end{algorithm}
\begin{figure}[b]
  \centering
  \begin{tikzpicture}[node distance=25pt, thick, font=\footnotesize,
    main/.style={draw, circle, minimum size=11pt, fill=cyan, inner sep=1.9pt, text width=6pt},
    rnn/.style={draw, rectangle, minimum size=11pt, fill=magenta}]
    % Time 1
    %% Nodes
    \node[main] (111) {$y_1$};
    \node[rnn] (121) [above right of=111] {$z$};
    \node[main] (112) [below right of=121] {$y_2$};
    \node[main] (101) [below of=111] {$x_1$};
    \node[main] (102) [below of=112] {$x_2$};
    %% Edges
    \draw[->] (111) -- node[midway, above left, sloped, pos=.8] {1} (121);
    \draw[->] (112) -- node[midway, above right, sloped, pos=.8] {-1} (121);
    \draw[->] (101) -- node[midway, left] {-1} (111);
    \draw[->] (101) -- node[midway, above, sloped, pos=.85] {1} (112);
    \draw[->] (102) -- node[midway, above, sloped, pos=.85] {1} (111);
    \draw[->] (102) -- node[midway, right] {1} (112);
    % Time 2
    %% Nodes
    \node[main] (211) [right of=112, node distance=33pt] {$y_1$};
    \node[rnn] (221) [above right of=211] {$z$};
    \node[main] (212) [below right of=221] {$y_2$};
    \node[main] (201) [below of=211] {$x_1$};
    \node[main] (202) [below of=212] {$x_2$};
    %% Edges
    \draw[->] (121) -- node[midway, above] {1} (221);
    \draw[->] (211) -- node[midway, above left, sloped, pos=.8] {1} (221);
    \draw[->] (212) -- node[midway, above right, sloped, pos=.8] {-1} (221);
    \draw[->] (201) -- node[midway, left] {-1} (211);
    \draw[->] (201) -- node[midway, above, sloped, pos=.85] {1} (212);
    \draw[->] (202) -- node[midway, above, sloped, pos=.85] {1} (211);
    \draw[->] (202) -- node[midway, right] {1} (212);
    % Time 3
    %% Nodes
    \node[main] (311) [right of=212, node distance=33pt] {$y_1$};
    \node[rnn] (321) [above right of=311] {$z$};
    \node[main] (312) [below right of=321] {$y_2$};
    \node[main] (301) [below of=311] {$x_1$};
    \node[main] (302) [below of=312] {$x_2$};
    %% Edges
    \draw[->] (221) -- node[midway, above] {1} (321);
    \draw[->] (311) -- node[midway, above left, sloped, pos=.8] {1} (321);
    \draw[->] (312) -- node[midway, above right, sloped, pos=.8] {-1} (321);
    \draw[->] (301) -- node[midway, left] {-1} (311);
    \draw[->] (301) -- node[midway, above, sloped, pos=.85] {1} (312);
    \draw[->] (302) -- node[midway, above, sloped, pos=.85] {1} (311);
    \draw[->] (302) -- node[midway, right] {1} (312);
    % Time inf
    \node (421) [right of=321, node distance=40pt] {$\bm{\cdots}$};
    \draw[->] (321) -- (421);
  \end{tikzpicture}
  \Description{A Recurrent Neural Network (RNN) with one feed-forward layer of size two and one recurrent layer of size one.}
  \caption{A simple RNN, consisting of a feed-forward first layer
    \(\vy = \ReLU(\vomega_1\vx)\), which maps the input
    \(\vx=\begin{psmallmatrix}x_1\\x_2\end{psmallmatrix}\) to
    \(\vy=\begin{psmallmatrix}y_1\\y_2\end{psmallmatrix}\) via
    left multiplication by
    \(\vomega_1=\begin{psmallmatrix}-1&1\\1&1\end{psmallmatrix}\). The
    last layer is recurrent and is defined as
    \(z^{\tup{t}} = \ReLU(\vomega_2\vy + z^{\tup{t-1}})\), where
    \(\vomega_2=\begin{psmallmatrix}1&-1\end{psmallmatrix}\) and \(z\)
    is the output variable with \(z^{\tup{0}} = 0\). The superscript
    \(^{\tup{t}}\) denotes indicates the time step and is omitted when
    there is no ambiguity (e.g., here, we have omitted it for
    \(x_1, x_2, y_1\) and \(y_2\), as they are not used in other time
    steps).}
  \label{fig: PWL}
\end{figure}
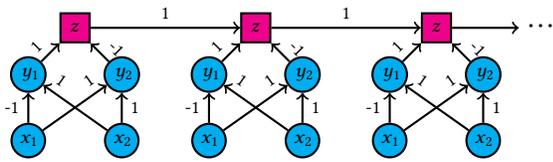
\begin{example}
  \label{ex: PWL}
  Consider the formula \(\phi_i = \Next^i (z \leq 2) \in\) \LTL, where
  \(z\) is the output of the RNN in \cref{fig: PWL}, with its input
  being \(\vx \in [0, 1]^2\). We want to verify \(\phi_2\) and
  \(\phi_3\) are satisfied.
\end{example}
Let us check whether \(\phi_2\) and \(\phi_3\) are satisfied using
\cref{alg: BPT}. At \(t = 0\), we have that
\([\vell^{\tup{0}}, \vu^{\tup{0}}] = [0, 1]^2 \times [0, 0]\), where
\([0, 1]^2\) indicates the bounds on \(\vx^{\tup{0}}\) and \([0, 0]\)
indicates the bounds on \(z^{\tup{0}}\) (line 1).
To obtain bounds on \(z^{\tup{1}}\), we need to propagate the bounds
through the RNN, which gives us
\(y^{\tup{1}}_1 = \ReLU(-x^{\tup{0}}_1 + x^{\tup{0}}_2) \in [0,
1]\). Similarly, we have \(y^{\tup{1}}_2 \in [0,2]\) and
\(z^{\tup{1}} \in [0, 1]\); hence,
\([\vell^{\tup{1}}, \vu^{\tup{1}}] = [0, 1]^3\). Repeating this, we
have \([\vell^{\tup{2}}, \vu^{\tup{2}}] = [0, 1]^2 \times [0, 2]\) and
\([\vell^{\tup{3}}, \vu^{\tup{3}}] = [0, 1]^2 \times [0,3]\); thus,
using BPT, we could verify that \(\phi_2\) is satisfied while \(\phi_3\) is not (since \(3 \not\leq 2\)).

\subsubsection{Recursive Mixed Integer Linear Programming (RMILP)}
At each step, BPT overapproximates the set of attainable values for
each variable. Hence, it does not provide a complete method
despite its efficiency. 
In RMILP, instead of propagating the bounds
for each variable, we construct an MILP equation describing the
set of attainable values \emph{exactly} in terms of the constraints
and variables of previous time steps, in a \emph{recursive} manner. Similarly to
BPT, it suffices to verify formulae of the form
\(\Next^i (\tr{\vc} \cdot \vx \leq d)\).

As outlined in \cref{alg: RMILP}, we want to verify that for all $\vx \in \sP_{\sX}^{\tup{i}}$, \(\tr{\vc} \cdot   \vx \leq d\), where 
\(\sP_{\sX}^{\tup{i}}\) is the joint system and hidden state at time
\(i\). 
This is verified using MILP by first finding the $\vx^*\in \sP_{\sX}^{\tup{i}}$ that maximises \(\tr{\vc} \cdot \vx\) and then checking whether \(\tr{\vc} \cdot \vx^* \leq d\). 
However, we first need to build the constraints that define \(\sP_{\sX}^{\tup{i}}\). We build them recursively, as described in the subprocedure \texttt{buildConstraints} of \cref{alg: RMILP} (lines 4-11). For any given step $t$, it returns the set of MILP variables $\sV$ and constraints $\sC$ that encode the reachable set \(\sP_{\sX}^{\tup{t}}\). For simplicity, we denote the MILP variables using $\vx^{\tup{t}}$ and $\vh^{\tup{t}}$. 

Returning to \cref{ex: PWL}, recall that we could not check
whether \(\phi_3\) is satisfied using BPT. Let us instead use RMILP to
check \(\phi_3\). We have that \(\phi_3\) holds iff
\(z^{\tup{3}} \leq 3\). Since
\(z^{\tup{3}} = \ReLU(\vomega_2 \ReLU(\vomega_1\vx^{\tup{2}}) +
z^{\tup{2}})\), \(z^{\tup{3}}\) is defined in terms of \(z^{\tup{2}}\)
and the input \(\vx^{\tup{2}}\).
% \np{in the example we use $y$ not $\vx^{\tup{3}}$, double check}
% \mh{$z$ is a function of $y$ and $y$ is a function of $x$...}
In a similar manner, we can write \(z^{\tup{2}}\) in terms of
\(z^{\tup{1}}\) and \(\vx^{\tup{1}}\), and \(z^{\tup{2}}\) in terms of
\(z^{\tup{0}}\) and \(\vx^{\tup{0}}\). This results in PWL constraints
defining \(z^{\tup{3}}\) in terms of
\(\vx^{\tup{0}}, \vx^{\tup{1}}, \vx^{\tup{2}}\) and
\(z^{\tup{0}}\),\footnote{We note that \(\ReLU(x) = \max\set{0, x}\)
  can be written as a PWL constraint using the ``big-M'' method
  \cite{GNS09} as follows.
  \(y = \ReLU(x) \iff (y \geq x) \wedge (y \geq 0) \wedge (y = \delta
  x)\).}  allowing us to check whether \(z^{\tup{3}} \leq 3\), using
MILP. As we show next, RMILP allows us to prove that \(\phi_3\), in
fact, holds.

\subsubsection{Adaptive Splitting}
% \label{subsec: Adaptive Splitting}
One limitation of RMILP is its high computational complexity, as we prove in
\cref{thm: Bounded}, leading to an exponentially large search
space for MILP solvers.
% \np{not immediately clear what's the connection between \coNPComplete{ness} and the text below. shall we add that, in the worst case, this implies that the complexity of the algorithm is exponential in the number of discrete variables?}
% \np{I thought we were using BPT to inform case splitting. Let's
% make sure this is clear in the text} \mh{done}
We use adaptive splitting
\cite{KouvarosL21,HenriksenL20} to mitigate this problem. In adaptive
splitting, we use BPT to provide the MILP solver with \textit{a priori} stricter bounds for
the continuous variables, thereby reducing the number of
branchings required to solve the MILP problem; i.e., the BPT bounds allow us to restrict the values that the discrete variables
in \cref{eq: PWL} attain. 

In \cref{ex: PWL}, we showed that \(x_1+x_2 \in [0, 2]\)
using BPT (for all \(t \in \NN\), hence, dropping the suffix \(^{\tup{t}}\)). So, \(y_2 = \ReLU(x_1 + x_2)\) simplifies to \(y_2= x_1 + x_2\), allowing us
to reduce the number of PWL constraints. 
Thus, adaptive splitting simplifies the equation
for \(z\) to
\(z^{\tup{t}} = \ReLU(\ReLU(-x_1 + x_2) - (x_1 + x_2) +
z^{\tup{t-1}})\). We can solve the resulting constraints as follows: for \(t = 1\), if \(x_2 \geq x_1\), then
\(z^{\tup{1}} = \ReLU(-2x_1) = 0\), and if \(x_2 < x_1\), then
\(z^{\tup{1}} = \ReLU(-x_1 - x_2) = 0\) for all
\(x_1, x_2 \in [0, 1]\). Repeating this for the following steps, we
obtain that \(z^{\tup{t}} = 0\) for all \(t\). Thus, not only
\(\phi_2\), but also \(\phi_3\) holds. 

In summary, the approach uses RMILP for verification, but first applies adaptive splitting to reduce any $\ReLU(y)$ expression into either $y$ or $0$ whenever BPT can show that $y\geq 0$ or $y\leq 0$, respectively. Similar reduction rules can be defined for other PWL activation functions, such as LeakyReLU, HardShrink and SoftShrink.

%%% Local Variables:
%%% mode: latex
%%% TeX-master: "../main"
%%% End:

% 
%%
\section{\sysname{'} Unbounded \LTL Verification}
\label{sec: Unbounded Verification}
In this section, we address the UVP for \sysname{} and show how to
solve it for a wide range of unbounded \LTL{} specifications using
\textbf{(1)} \emph{Bounded-Path Verification} (BPMC), \textbf{(2)}
\emph{lasso} search, and \textbf{(3)} \emph{inductive invariant}
synthesis. Before presenting our decidability results, we note that
the UVP is undecidable in general (proof in \cref{sec: Undecidability Proof}).
%
% \vspace{-0.1em}
\begin{theorem}
  \label{thm: Undecidability}
  \LTL{} verification of \sysname{} is undecidable.
\end{theorem}
% \end{tcolorbox}

%
%%
% \vspace{-0.6em}
\subsection{Bounded-Path Model Checking}
\emph{Bounded-Path Model Checking} (\emph{BPMC}) allows us to
partially solve the UVP by restricting unbounded formulae to finite
paths. It is a straightforward, yet effective, technique to solve the
UVP as we formally state in \cref{thm: BMC for UVP}.
\begin{theorem}
  \label{thm: BMC for UVP}
  Given formulae \(\phi, \psi \in\) \LTL{}, such that
  \(S \models \phi\) and \(S \models \psi\) are decidable, we
  can solve the instances of the UVP checkmarked in \cref{tbl: All
    Approaches}, using BPMC.
\end{theorem}
\begin{proof}
  To verify the specifications in \cref{tbl: All Approaches}
  using BPMC, it suffices to write their bounded versions using
  \(\Next\) and Boolean operators. To this end, we observe that for a
  given \(k \in \NN\), we have that
  \(\Finally^{\leq k} \phi \equiv \bigvee_{i=1}^{k} \Next^{i} \phi\),
  \(\Always^{\leq k} \phi \equiv \bigwedge_{i=1}^{k} \Next^{i} \phi\),
  \(\psi \Until^{\leq k} \phi \equiv \bigvee_{i=0}^k
  ((\bigwedge_{j=0}^{i-1} \Next^j \psi) \wedge \Next^i \phi)\), and
  \(\psi \Release^{\leq k} \phi \equiv \bigvee_{i=0}^{k-1}
  ((\bigwedge_{j=0}^{i} \Next^j \psi) \wedge \Next^i \phi)\).
\end{proof}
\begin{algorithm}[t]
  \caption{Lasso-assisted BPMC for \(\Finally \phi\)}
  \label{alg: Verifying Finally}
  \DontPrintSemicolon
  \Input{\sysname{} \(S\), \(\phi \in \) \LTL, and search depth
    \(k\)} \Output{\texttt{False/Unknown}}

  \(\sP^{\tup{0}} = \set{(\vx, \vh) \in \sX^{\tup{0}} \times
    \sH^{\tup{0}} : \vx \not\models \phi}\)

  \For{\(t \gets 1, \dots, k\)}{

    % \(\sP^{\tup{t}} = \set{(\vx', \textbf{h}') \in \sX \times \sH :
    %   \vx' \not\models \phi \wedge
    %   \exists (\vx, \vh) \in \sP^{\tup{t-1}},
    %   \exists \mathbf{a} \in \sA (
    %   \mathbf{a} \!=\! r(o(\vx), \vh),
    %   (\vx' \!\in\! \tau(\vx, \mathbf{a}))}\)

    \(\sP^{\tup{t}} = \set{(\vx, \vh) \in e(\sP^{\tup{t-1}}) : \vx \not\models \phi}\)
    
    \For{\(t' \gets 0, \dots, t-1\)}{
      
      \If{\(\sP^{\tup{t}} \cap \sP^{\tup{t'}} \neq \emptyset\)}{
        \Return \False \Comment*[r]{Lasso is found}
      }
    }
  }
  \Return \Unknown
\end{algorithm}
%

%
%%
% \vspace{-0.2em}
\subsection{Lasso Search}
\label{subsec: Lasso Search}
Bounded model checking for solving the UVP can be further enhanced via
lasso search \cite{Biere+09}. For an \sysname{} \(S\) and formula
\(\phi \in\) \LTL, the idea behind the lasso search is to find a path
\(\rho \in \Pi_S\), such that for some 
\(0 \leq t' < t\), we have that
\(\rho^{\tup{1}}, \dots, \rho^{\tup{t}} \not\models \phi\) and
\(\rho^{\tup{t'}} = \rho^{\tup{t}}\). Such a path is called a
\emph{lasso}, and \(\rho^{\tup{1, \dots, t}}\) is witness for
\(S \not\models \Finally \phi\). This is summarised in
\cref{alg: Verifying Finally}.

We observe that \cref{alg: Verifying Finally} is complete.
To prove this, we note that it returns \False{} only if
\(\sP^{\tup{t}} \cap \sP^{\tup{t'}} \neq \emptyset\) for some
\(0 \leq t' < t \leq k\). This means that there exist
\((\vx, \vh) \in \sP^{\tup{t}} \cap \sP^{\tup{t'}}\),
\(\vx^{\tup{0}}, \dots, \vx^{\tup{t}} \in \sX\), and
\(\vh^{\tup{0}}, \dots, \vh^{\tup{t}} \in \sH\), such that
\(\vx^{\tup{t'}} = \vx^{\tup{t}} = \vx\),
\(\vh^{\tup{t'}} = \vh^{\tup{t}} = \vh\), and
\((\vx^{\tup{i}}, \vh^{\tup{i}}) = \tau(\vx^{\tup{i-1}},
r(o(\vx^{\tup{i-1}}), \vh^{\tup{i-1}}))\) for \(i = 1, \dots,
t\). This implies that \(\vx^{\tup{0}}, \dots, \vx^{\tup{t}}\) is a
lasso that entirely does not satisfy \(\phi\); thus,
\(S \not\models \Finally \phi\), and \cref{alg: Verifying
  Finally} is complete. \cref{thm: Lasso} generalises this
reasoning.
\begin{theorem}
  \label{thm: Lasso}
  Given \LTL{} formulae \(\phi\) and \(\psi\), such that
  \(S \models \phi\) and \(S \models \psi\) are solvable (e.g., using
  BPMC, lasso, or invariants), then, using lasso search, we can solve
  the instances of UVP checkmarked in \cref{tbl: All Approaches}.
\end{theorem}
\begin{proof}
  We have already shown in \cref{subsec: Lasso Search} that the lasso
  search allows verifying \(S \not\models \Finally\phi\).  To check
  \(S \not\models \psi \Until \phi\), we first search for a lasso for
  \(S \not\models \Finally\phi\). If a lasso is not found, we cannot
  prove the specification. Otherwise, we have found a lasso of length
  \(k\) proving that \(S \not\models \Finally\phi\). This means that
  \(S\) will never satisfy \(\phi\); hence
  \(S \not\models \psi \Until \phi\). 
  % A similar approach can be used
  % for checking \(S \not\models \psi \Release \phi\).
  %
  To check \(S \not\models \psi \Release \phi\), we first search for a
  lasso for \(S \not\models \Finally\psi\). If a lasso is not found,
  we cannot prove the specification. Otherwise, we have found a lasso
  of length \(k\) proving that \(S \not\models \Finally\psi\). This
  means that \(S\) will never satisfy \(\psi\). Now, for a fixed bound
  \(i \in \NN\), we check whether
  \(S \not\models \Next^{\leq i} \phi\). If
  \(S \not\models \Next^{\leq i} \phi\), we have that
  \(S \not\models \psi \Release \phi\).
  % ToDo: Add \Release proof.
\end{proof}
  \begin{table}[t]
    \centering
    \caption{Specifications that some (not all) instances of that form
    \emph{may} be verified using each
      method. A cross mark \xmark{} means that an approach (e.g.,
      Lasso) \emph{cannot} be used for instances of that form (e.g.,
      \(\psi \Release \phi\)), and a checkmark \cmark{} means it
      \emph{can}, however, it may return the correct answer or
      \Unknown.}
    \label{tbl: All Approaches}
    % \vspace{-0.4em}
    \begin{tabular}{lccccc}
      % \toprule
      Method & Satisfaction & $\psi \Release \phi$ & $\Always\, \phi$  & $\psi \Until \phi$ & $\Finally\, \phi$\\
      \midrule
      \multirow{2}{*}{BPMC} & $S \models$ & \cmark & \xmark & \cmark & \cmark\\
      & $S \not\models$ & \cmark & \cmark & \cmark & \xmark\\
      \midrule
      \multirow{2}{*}{Lasso} & $S \models$ & \xmark & \xmark & \xmark & \xmark\\
      & $S \not\models$ & \cmark & \xmark & \cmark & \cmark\\
      \midrule
      \multirow{2}{*}{Invariant} & $S \models$ & \cmark & \cmark & \xmark & \xmark\\
      & $S \not\models$ & \xmark & \xmark & \cmark & \cmark\\
      % \bottomrule
    \end{tabular}
    % \vspace{-0.4em}
  \end{table}
\begin{remark}
  \label{rem: Lasso}
  \cref{alg: Verifying Finally} can only be used alongside
  MILP because BPT overapproaximates \(\sP^{\tup{t}}\), and therefore
  the algorithm may wrongly return \False{} in line 5.
\end{remark}
\subsection{Inductive Invariant Synthesis}
Despite the versatility of the lasso approach, it cannot be used to
verify specifications of the form \(S \models \Always\phi\) in
non-deterministic environments. This is because to verify
\(S \models \Always\phi\), we need to ensure that
\(\pi \models \Always \phi\) for \emph{all} valid paths
\(\pi \in \Pi_{S}\), whilst lassos only allow us to search for the
\emph{existence} of a path that satisfies a given specification.
This leads us to use inductive invariants \cite{Biere+09}. For an
\sysname{} \(S\), an \emph{inductive invariant}, or \emph{invariant}
for short, is a set \(\sI \subseteq \sX \times \sH\), such that
\(\sX^{\tup{0}} \times \sH^{\tup{0}} \subseteq \sI\) and
\(e(\sI) \subseteq \sI\). If there exists an inductive invariant
\(\sI\) for \(S\) such that \(\sI \models \phi\) for a formula
\(\phi \in\) \LTL, it follows that \(S \models \Always\phi\) since all
initial and future states of \(S\) satisfy \(\phi\). This allows us to
verify a wider range of unbounded formulae, as shown in
\cref{thm: Inductive Invariants} and outlined in
\cref{tbl: All Approaches}
\begin{theorem}
  \label{thm: Inductive Invariants}
  Given an \sysname{} \(S\) and \LTL{} formulae \(\phi\) and \(\psi\),
  such that \(S \models \phi\) and \(S \models \psi\) are solvable,
  then, using inductive invariants, we can solve the instances of UVP
  checkmarked in \cref{tbl: All Approaches}.
\end{theorem}
\begin{proof}
  We show that inductive invariants can be used to prove
  \(S \models \Always \phi\) specifications. Assume that
  \(\sI\) is an inductive invariant for \(S\). If
  \(\sI_{\sX} \models \phi\), then by definition, we have that the
  initial states and all future states resulting from finite
  evolution of \(S\) satisfy \(\phi\). Hence, an invariant \(\sI\) for
  \(S\), such that \(\sI_{\sX} \models \phi\), implies
  \(S \models \Always \phi\).
  To prove that inductive invariants can be used to verify
  \(S \not\models \Finally \phi\) and
  \(S \not\models \psi \Until \phi\), note that we can check for
  \(S \models \Always \neg \phi\) using inductive invariants. If we can
  verify that \(S \models \Always \neg \phi\), then
  \(S \not\models \Finally \phi\) and
  \(S \not\models \psi \Until \phi\).
  A similar reasoning applies to
  \(S \models \psi \Release \phi\). We can check
  \(S \models \Always\phi\) using invariants. If we can
  verify that \(S \models \Always\phi\), then
  \(S \models \psi \Release \phi\).
\end{proof}
\begin{algorithm}[t]
  \caption{Inductive invariants for \(\Always \phi\)}
  \label{alg: Verifying Always}
  \DontPrintSemicolon

  \Input{\sysname{} \(S\), \(\phi \in \) \LTL, and search depth \(k\)}
  \Output{\texttt{True/False/Unknown}}

  \(\sI^{\tup{0}} = \sX^{\tup{0}} \times \sH^{\tup{0}}\)

  \(\sM^{\tup{0}} = \sX \times \sH\)

  \For{\(t \gets 1, \dots, k\)}{

    \(\sI^{\tup{t}} = e(\sI^{\tup{t-1}})\)

    \(\sM^{\tup{t}} = e(\sM^{\tup{t-1}})\)

    \uIf{\(\sI_{\sX}^{\tup{t}} \not\models \phi\)}{
      \Return \False
    }
    \uElseIf{\(\sM_{\sX}^{\tup{t-1}} \models \phi\)}{      
      \Return \True \Comment*[r]{Maximal invariant}
    }
    \ElseIf{\(\sI^{\tup{t}} \subseteq \bigcup_{i=0}^{t-1} \sI^{\tup{i}}\)}{
      \Return \True \Comment*[r]{Minimal invariant}
    }
  }
  \Return \Unknown
\end{algorithm}

Even though inductive invariants theoretically allow us to verify a
wider range of unbounded formulae, finding such invariants in practice
is not straightforward. We provide two techniques for this purpose in
\cref{alg: Verifying Always} by searching for \emph{maximal}
and \emph{minimal} inductive invariants. For the minimal invariant, we
start from \(\sI = \sI^{\tup{0}} =\sX^{\tup{0}} \times \sH^{\tup{0}}\)
(line 1), which must be contained by any invariant by definition. We
then keep expanding \(\sI\) by
\(\sI^{\tup{t}} = e^{(t)}(\sI^{\tup{0}})\), until we find an
invariant (line 10), in which case
\(\sI = \bigcup_{i=0}^{t-1} \sI^{\tup{i}}\) is the minimal invariant, or
we exceed a given maximum number of iterations \(k\).  For the maximal invariant,
we start with the trivial largest invariant, i.e.,
\(\sM^{\tup{0}} = \sX \times \sH\). If
\(\sM_{\sX}^{\tup{0}} = \sX \models \phi\), it follows that
\(S \models \Always \phi\). Otherwise, the algorithm iteratively
computes \(\sM^{\tup{t}} = e^{(t)}(\sM^{\tup{0}})\) until
\(\sM_{\sX}^{\tup{t}} \models \phi\) (line 8), in which case
\(\sM = \bigcup_{i=0}^{t-1} \sI^{\tup{i}} \cup \sM^{\tup{t-1}}\) is the 
invariant, or it exceeds a given maximum time step \(k\).

To prove that \(\sI\) and \(\sM\) (if found) in \cref{alg:
  Verifying Always} are invariant, we observe that by construction,
\(\sX^{\tup{0}}\times\sH^{\tup{0}} \subseteq \sI, \sM\). For \(\sI\),
we have that
\(e(\sI) = e(\bigcup_{i=0}^{t-1} \sI^{\tup{i}}) = \bigcup_{i=0}^{t-1}
e(\sI^{\tup{i}}) = \bigcup_{i=0}^{t-1} \sI^{\tup{i+1}} \subseteq
\bigcup_{i=0}^{t-1} \sI^{\tup{i}} = \sI\) by line 10; therefore,
\(\sI\) is an invariant. For \(\sM\), we have that
$\sI^{\tup{t}}, \sM^{\tup{t}} \subseteq
\sM^{\tup{t-1}}=e^{t-1}(\sM^{\tup{0}})$ because
$\sI^{\tup{t}}=e^{t-1}(\sI^{\tup{1}})$ and
$\sM^{\tup{t}}= e^{t-1}(\sM^{\tup{1}})$, and by construction, it follows that
\(\sI^{\tup{1}}, \sM^{\tup{1}} \subseteq \sM^{\tup{0}}\).
Hence,
\(e(\sM) = e(\bigcup_{i=0}^{t-1} \sI^{\tup{i}} \cup \sM^{\tup{t-1}}) =
\bigcup_{i=0}^{t-1} e(\sI^{\tup{i}}) \cup e(\sM^{\tup{t-1}}) =
\bigcup_{i=0}^{t-1} \sI^{\tup{i+1}} \cup \sM^{\tup{t}} \subseteq
\bigcup_{i=0}^{t-1} \sI^{\tup{i}} \cup \sM^{\tup{t-1}} = \sM\);
\mbox{thus, \(\sM\) is an invariant.}

\vspace{2pt}
\begin{remark}
  \label{rem: Inductive Invariant}
  \cref{alg: Verifying Always} works with both RMILP and BPT.
  However, since BPT is incomplete, we have to return \Unknown{}
  instead of \False{} in line 7 when using BPT (because \(\phi\) may be violated by a spurious states in \(\sI_{\sX}^{\tup{t}}\).)
\end{remark}
%

%%% Local Variables:
%%% mode: latex
%%% TeX-master: "../main"
%%% End:

%
%%
\section{Experimental Evaluations}
\label{sec: Evaluation}
We evaluate\footnote{For all experiments, training and verification
  were performed on a workstation with a 16-core AMD 5955WX CPU, 256GB
  of DDR4 RAM, and an NVIDIA RTX 4090 GPU. The verification algorithms
  are implemented in Python using the NumPy library \cite{Numpy} and
  Gurobi optimiser \cite{Gurobi}. The source code will be released
  on \href{https://github.com/mehini/Vern}{\texttt{github.com/mehini/Vern}}.}
  our approach by verifying unbounded and
bounded \LTL{} specifications for the Cart Pole, Pendulum, and Lunar
Lander environments from the Gymnasium library \cite{Gymnasium}, and
the Simple Push environment from the Petting Zoo library
\cite{PettingZoo}.  For each environment, we consider memoryful RNN
policies of varying complexity. We use \(R_{i}\) and \(R_{2\times i}\)
to denote neural policies with one and two recurrent layers of size
\(i\), respectively. These recurrent layers are followed by single
linear layers matching the agents' action spaces. We denote with
\(S_{i}\) and \(S_{2\times i}\) the corresponding \sysname{}.

For \sysname{} with discrete action spaces, the policies are trained
using recurrent Q-learning (DQN) \cite{Mnih+13,HausknechtS15}, while
for continuous action spaces, we use Proximal Policy Optimisation
(PPO) \cite{Schulman+17}.  The environments' transition functions are
non-linear and involve trigonometric functions; hence, we use PWL
approximation of these \cite{AkitundeLMP18,DAmbrosioLM10} to make them
linearly definable.  In all experiments, our implementation first
attempts to verify using BPT; if unsuccessful, it uses RMILP with
adaptive splitting. We have set a timeout of 10,000 seconds. In all
tables, ``\(-\)'' signifies timeout.

%
% \begin{figure}[t]
%   \centering
%   %
%   \begin{subfigure}[b]{0.242\linewidth}
%     \centering
%     \includegraphics[width=\textwidth]{cartpole-2}
%     % \caption*{\scriptsize (a) Cartpole}
%     % \label{subfig: Cartpole}
%   \end{subfigure}
%   % 
%   \hfill
%   %
%   \begin{subfigure}[b]{0.242\linewidth}
%     \centering
%     \includegraphics[width=\textwidth]{pendulum-2}
%     % \caption*{\scriptsize (b) Pendulum}
%     % \label{subfig: Pendulum}
%   \end{subfigure}
%   % 
%   \hfill
%   %
%   \begin{subfigure}[b]{0.242\linewidth}
%     \centering
%     \includegraphics[width=\textwidth]{lunar_lander-2}
%     % \caption*{\scriptsize (c) Lunar Lander}
%     % \label{subfig: Lunar Lander}
%   \end{subfigure}
%   % 
%   \hfill
%   %
%   \begin{subfigure}[b]{0.242\linewidth}
%     \centering
%     \includegraphics[width=\textwidth]{simple_spread-2}
%     % \caption*{\scriptsize (d) Simple Spread}
%     % \label{subfig: Simple Spread}
%   \end{subfigure}
%   %
%   \caption{The four environments considered in the paper.}
%   %Environments (b)-(d) have
%   %  non-deterministic initial states, and (c) has non-deterministic
%   %  transitions.}
%   \label{fig: Environments}
% \end{figure}
%

%
%%
\subsection{Cart Pole (Verifying Always)}
In the Cart Pole environment \cite{BartoSA83}, a pole is attached by
an unactuated joint to a cart that is moving on a frictionless
track. The pole is placed upright on the cart, and the goal is to
balance the pole by applying forces to the cart from the left and
right. At each time step, the agent has access to the cart's position
\(x \in [-4.8, 4.8]\) and the pole's angle
\(\theta \in [\nicefrac{-2\pi}{15}, \nicefrac{2\pi}{15}]\) and can
choose between pushing the cart to the left or right. We want to verify that
the pole always remains balanced; thus, we consider unbounded
specifications
\(\phi_{\epsilon} = \Always ((\theta \leq \epsilon) \wedge (|x| \leq
1))\), for given angle bounds $\epsilon$, when starting from the
centre of the plane \(x = 0\) while the pole is vertical, i.e.,
\(\theta = \nicefrac{\pi}{2}\). We trained three policies
\(R_{16}, R_{32}\), and \(R_{2\times 16}\), and used \cref{alg:
  Verifying Always} to check whether their corresponding environments
\(S_{16}, S_{32}\), and \(S_{2 \times 16}\) satisfy
\(\phi_{\epsilon}\) for different \(\epsilon\). As reported in
\cref{tbl: Cart Pole Always}, the \(R_{2 \times 16}\)
policy is guaranteed to keep the pole balanced indefinitely, and
instances are easily verified, while \(R_{16}\) does not satisfy any
of the constraints and the \(S_{16}\) environment is easily
falsified. On the other hand, verifying the \(S_{32}\) environment has
proven to be more challenging.

\begin{table}[t]
  \centering
  \caption{Runtimes (seconds) for \cref{alg: Verifying
      Always}, with a maximum search depth of \(9\), to verify
    \(S \models \phi_\epsilon\) for different \(\epsilon\). Green,
    white, and red cells, denote that \cref{alg: Verifying
      Always} returns \True, \Unknown, and \False, respectively. All
    specifications are verified using RMILP. Superscripts \dag{} and
    \ddag{} indicate that minimal and maximal invariants are found,
    respectively. Subscripts denote the search depth where
    \cref{alg: Verifying Always} has terminated.}
  \label{tbl: Cart Pole Always}
  \vspace{-0.3em}
  % \resizebox{\linewidth}{!}{
  \begin{tabular}{rrrrr}
    % \toprule
    % & \multicolumn{3}{c}{\(\epsilon\)} \\
     & $\epsilon=\pi/40$          & $\epsilon=\pi/30$           & $\epsilon=\pi/20$           & $\epsilon=\pi/10$\\
    % \cmidrule{2-4}
    \midrule
    $S_{16} \models \phi_{\epsilon}$        & \celluns $2.7_2$         & \celluns $2.3_1$         & \celluns $2.0_1$         & \celluns $2.3_1$\\
    $S_{32} \models \phi_{\epsilon}$        & \cellsat $41.3_3^{\dag}$ & $-_9$                    & $7571.2_8$               & $-_9$\\
    $S_{2\times16} \models \phi_{\epsilon}$ & \cellsat $9.8_1^{\ddag}$ & \cellsat $5.3_1^{\ddag}$ & \cellsat $9.3_1^{\ddag}$ & \cellsat $11.7_2^{\dag}$\\
    %\bottomrule
  \end{tabular}
  % }
  \vspace{-0.3em}
\end{table}
%

% \balance

%
%%
\subsection{Pendulum (Verifying Eventually)}
% \np{here (and in other case studies where relevant) do not write ``random position'' but e.g., the pendulum's initial state is uncertain and can take any value in \ldots}
A pendulum is attached to a fixed joint at one
end, while the other end is free. The pendulum starts in an uncertain
position, and the goal is to swing it
to an upright vertical position by applying torque on the free end. At each step, the agent has
access to the free end's horizontal and vertical positions \((x, y)\)
and can apply a torque \(a \in [-2, 2]\).

\begin{table}[t]
  \centering
  \caption{Runtimes (seconds) of \cref{alg: Verifying
      Finally} (\False{} and \Unknown{} instances) and BPMC (\True{}
    instances), with a maximum search depth of \(20\), for verifying
    \(S \models \psi_\epsilon\) for different \(\epsilon\). Subscripts
    denote the time steps where violations are found in \False{}
    instances and the search depths \(k\) of \cref{alg:
      Verifying Finally} in other instances. The colour code is as in
    \cref{tbl: Cart Pole Always}. All instances are verified by
    first using BPT and then RMILP; however, BPT alone could not verify
    any of the instances. For \Unknown{} instances, the table shows
    \cref{alg: Verifying Finally}'s time. Colour codes are the same as \cref{tbl: Cart Pole Always}.}
  \label{tbl: Pendulum Finally}
  \vspace*{-0.3em}
  % \resizebox{\linewidth}{!}{
  \begin{tabular}{rrrrr}
    % \toprule
    % & \multicolumn{3}{c}{\(\epsilon\)} \\
    & $\epsilon=0.15$          & $\epsilon=0.10$           & $\epsilon=0.05$           & $\epsilon=0.01$\\
    % \cmidrule{2-4}
    \midrule
    $S_{10} \models \psi_{\epsilon}$        & \celluns $60_{15}$   & \celluns $32_{13}$   & \celluns $86_{19}$ &  $183_{20}$\\
    $S_{20} \models \psi_{\epsilon}$        & \cellsat $2156_{19}$ & $6903_{20}$          &  $5314_{20}$       &  $6012_{20}$\\
    $S_{2\times10} \models \psi_{\epsilon}$ & \cellsat $8626_{18}$ & $-_{20}$             &  $-_{20}$          &  $-_{20}$\\
    %\bottomrule
  \end{tabular}
  % }
\end{table}

We verify that, for all possible system evolutions, the pole eventually reaches an upright position
\(y \geq 1 - \epsilon\) for different values of
\(\epsilon \in \set{0.15, 0.1, 0.05, 0.01}\), expressed by the
unbounded eventually property
\(\psi_{\epsilon} = \Finally (y \geq 1 - \epsilon)\). Again, we
consider three RNNs with varying numbers of layers and nodes. As
summarised in \cref{tbl: Pendulum Finally}, we see that
\(R_{10}\) has the worst performance among the three policies.

\subsection{Lunar Lander (Bounded Verification)}
\label{subsec: Lunar Lander}
The Lunar Lander environment \cite{LunarLander} consists of an agent
whose goal is to land on a landing pad between two
flags in a 2D representation of the moon
The environment is
non-deterministic due to wind.  The action space
\(\sA = \set{\circlearrowright, \uparrow, \circlearrowleft, \circ}\)
is discrete, consisting of ``fire right orientation engine'', ``fire
main engine'', ``fire left orientation engine'', and ``do
nothing''. The observation space \(\sO \subset \RR^5\) consists of the
horizontal and vertical coordinates of the lander \(x, y \in [0, 1]\),
its angle \(\theta \in [-\pi, \pi]\),
and two Boolean variables \(l_1, l_2 \in \set{0, 1}\) that indicate
if each leg is in contact with the ground.
The agent is rewarded for landing on the landing pad and is penalised
for crashing, moving out of the screen, and using the engines.

\begin{table}[b!]
  \centering
  \caption{Runtimes (seconds) for
    \(S_{2 \times 32} \models \varphi_{t, \epsilon} = \Next^t(|\theta|
    < \epsilon)\) for different values of \(t\) and \(\epsilon\) using
    the approach of (a) Akintunde et al., (b) Hosseini and Lomuscio,
    and (c) ours. Numbers indicate runtime (seconds). The white cells
    indicate that \(S \models \varphi_{t, \epsilon}\), whilst the
    red cells indicate that \(S \not\models \phi_{t,
      \epsilon}\). Instances marked by * are verified using BPT
    alone.}
  \label{tbl: Evaluating Bounded}
  \vspace{-0.3em}
  % \begin{subtable}[h]{\textwidth}
  \hspace*{-0.95em}
  \resizebox{1.04\linewidth}{!}{
    \begin{tabular}{crrrrrrrrrrr}
      $t$ & \(\pi/30\) & \(\pi/20\) & \(\pi/10\) &\(\pi/30\) & \(\pi/20\) & \(\pi/10\) & \(\pi/30\) & \(\pi/20\) & \(\pi/10\) \\
      \cmidrule(rl){2-4} \cmidrule(rl){5-7} \cmidrule(rl){8-10}
      % \multirow{ 10}{*}{\(t\)}
      1  & 23.1   & 22.4   & 25.0   & 5.8    & 6.7             &  6.1   & 0.0\rlap{$^*$} & 0.0\rlap{$^*$} & 0.1\rlap{$^*$}\\
      2  & 195.3  & 210.5  & 196.8  & 12.5   & 13.5            & 19.6   & 0.1\rlap{$^*$} & 0.1\rlap{$^*$} & 0.1\rlap{$^*$}\\
      3  & 413.8  & 437.1  & 440.4  & 25.6   & 27.1            & 28.9   & 0.2\rlap{$^*$} & 14.1           & 14.5          \\
      4  & 1728.7 & 1724.6 & 1829.4 & 49.1   & 55.0            & 59.6   & 33.3           & 37.4           & 38.7          \\
      5  & 8648.9 & 7972.0 & 8295.2 & 77.7   & 80.4            & 83.5   & 50.4           & 49.6           & 55.5          \\
      6  & $-$    & $-$    & $-$    & 92.9   & 99.7            & 113.5  & 73.1           & 72.3           & 80.0          \\
      7  & $-$    & $-$    & $-$    & 147.6  & 159.9           & 185.8  & 100.7          & 99.2           & 107.8         \\
      8  & $-$    & $-$    & $-$    & 335.9  & \celluns 393.9  & 340.2  & 193.4          & \celluns 287.4 & 211.0         \\
      9  & $-$    & $-$    & $-$    & 776.0  & \celluns 946.8  & 767.2  & 234.8          & \celluns 475.7 & 289.6         \\
      10 & $-$    & $-$    & $-$    & 1392.0 & \celluns 2241.2 & 1783.5 & 396.2          & \celluns 553.3 & 453.0         \\
      \\
      & \multicolumn{3}{c}{\small (a) RNSVerify \cite{Akintunde+19}}
      & \multicolumn{3}{c}{\small (b) RMILP \cite{HosseiniL23}}
      & \multicolumn{3}{c}{\small (c) Ours}
      % \bottomrule
    \end{tabular}
  }
\end{table}

Here, we trained an RNN policy \(R_{2 \times 32}\).
To compare our approach against the existing methods, we verified
\(S_{2 \times 32} \models \phi\)
using our proposed approach (RMILP combined with BPT and adaptive splitting)
as well as pure MILP-based approaches
of \cite{Akintunde+19} and \cite{HosseiniL23}. In fact, our implementation of \cite{HosseiniL23} is the same as pure RMILP, and thus, this also serves as an ablation study on how BP and adaptive splitting speed up the verification.

As we see in
\cref{tbl: Evaluating Bounded}, BPT allows us to verify some instances
in the early time steps almost instantly (the first rows of \cref{tbl: Evaluating Bounded}c); however, it fails to provide a definite outcome for larger
time bounds. Adaptive splitting significantly speeds
up the verification and in
larger time steps provides a speed up of around an order of magnitude
compared to the pure RMILP, which is also the
SoA \cite{HosseiniL23} runtime.
\subsection{Simple Spread (Multi-Agent)}
The Simple Spread environment \cite{PettingZoo} consists of \(n\)
agents and \(n\) landmarks. The goal is for the agents to cover all
landmarks while avoiding collisions. All agents have the same
observation space, \(\sO \subseteq \RR^{4n}\), consisting of all
agents and landmarks' 2D positions, and action space
\(\sA = \set{\leftarrow, \uparrow, \rightarrow, \downarrow, \circ}\),
consisting of ``move left'', ``move up'', ``move right'', ``move
down'', or ``do nothing'' actions.

We consider two environments with different numbers of agents
\(n = 2, 4\) and two composite (bounded) \LTL{} formulas: we verify whether the first agent, \(p_1\),
eventually reaches and remains in one of the landmarks (a stability property), and whether it always eventually reaches a landmark (a liveness property); i.e., we verify
\begin{equation*}
  \begin{split}
    S_n & \models \xi_t = \Always^{\leq t} \Finally^{\leq t}
          (\bigvee_{i=1}^n \norm{p_1 - \ell_i}_1 \leq 0.1),\\
    S_n & \models \chi_t = \Finally^{\leq t} \Always^{\leq t}
          (\bigvee_{i=1}^n \norm{p_1 - \ell_i}_1 \leq 0.1),
  \end{split}
\end{equation*}
where \(S_n\) is the \sysname{} with \(n\) agents and \(t=0,1,\ldots,9\).
% \(0 \leq t \leq 9\). 
The agents' initial states are uncertain and can take values in
\([-1, 1]^2\). Agents are controlled by \(R_{2 \times 8}\) policies. The
landmarks in \(S_2\) and \(S_4\) are at \((\pm\nicefrac{1}{2}, 0)\)
and \((\pm\nicefrac{1}{2}, \pm\nicefrac{1}{2})\). As
summarised in \cref{fig: Multi-Agent}, the
verification time grows almost exponentially in the number of time
steps \(t\), and except \(S_n \models \chi_t\), all other
specifications are unsatisfiable.

\begin{figure}[h!]
  \centering
  \includegraphics[width=.85\linewidth]{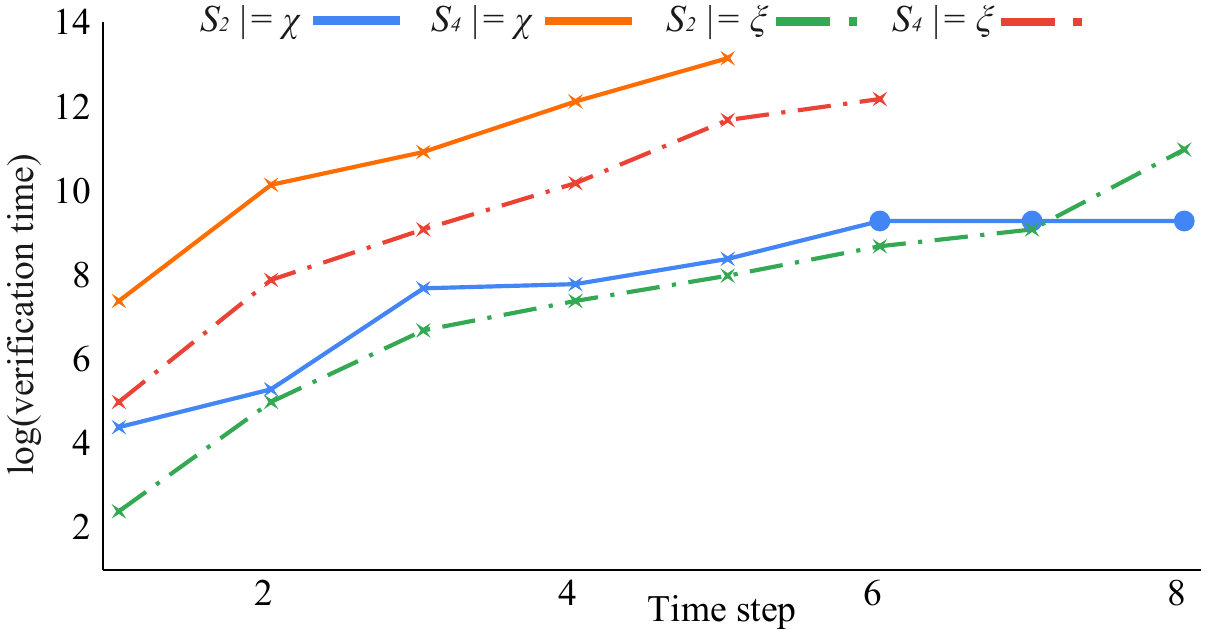}
  \vspace{-0.3em}
  \Description{Graphs of (logarithmic) verification time as the number
    of steps increase.}
  \caption{The base-2 logarithm of runtimes (seconds) for verifying
    \(\xi_t\) and \(\chi_t\) using BPMC. All instances are verified
    using RMILP. Symbols \({\color{green!80!blue!70!black} \bullet}\)
    and \xmark{} denote that the algorithm returned \True{} and
    \False, respectively.}
  \label{fig: Multi-Agent}
\end{figure}
\vspace{-1em}

%%% Local Variables:
%%% mode: latex
%%% TeX-master: "../main"
%%% End:

%
%%
\section*{Conclusions}
\label{sec: Conclusions}
We introduced the first framework for verifying full LTL specifications
of multi-agent systems with memoryful neural agents under
uncertainty by adapting well-established verification techniques.
We evaluated the introduced algorithms on
various widely used deep RL environments with single and multiple
agents and verified bounded and unbounded \LTL{} specifications. Compared to
the SoA that consider fragments of \LTL, we improved the verification time by an order of magnitude. This work demonstrates how foundational verification techniques can be used for verifying the safety of \sysname, thus paving the way for the adoption of other advanced verification methods for the safety verification of MAS.

%%% Local Variables:
%%% mode: latex
%%% TeX-master: "../main"
%%% End:

%
%%
\begin{acks}
This work was supported by supported by ``REXASI-PRO'' H-EU
project, call HORIZON-CL4-2021-HUMAN01-01, Grant agreement ID:
101070028 and by the Engineering and Physical Sciences Research Council
(EPSRC) under Award EP/W014785/2. Alessio Lomuscio acknowledges support from the Royal Academy of Engineering via a Chair of Emerging Technologies.

%%% Local Variables:
%%% mode: latex
%%% TeX-master: "../main"
%%% End:

\end{acks}

\bibliographystyle{ACM-Reference-Format} 
\balance
\bibliography{bibliography,sample}

\newpage
\begin{appendices}
\section{The Undecidability Proof}
\label{sec: Undecidability Proof}
In this section, we prove that the \LTL verification of \sysname{} is undecidable in general.

\begin{proof}[Proof of Theorem~\ref{thm: Undecidability}]
  We utilise the results of \cite{SiegelmannS95,Chen+18}, which show
  RNNs are universal Turing machines. Thus, to show that \LTL{}
  verification of \sysname{} is undecidable, it suffices to consider a
  fully-observable system \(S\) with a single agent controlled by an
  RNN policy \(R: (\RR^n)^* \to \RR^n\) alongside observation and
  actions spaces \(\sO = \sA = \RR^n\), identity observation function
  \(o : \vx \mapsto \vx\), and transition function
  \(\tau: (\vx, \va) \mapsto \va\). Now, by \cite[Theorem 8]{Chen+18},
  deciding whether \(S \models \Finally (a_1 > c)\) for a given
  \(0 < c < 1\) is undecidable.\footnote{We acknowledge that a similar
    undecidability result for memoryless neural MAS is provided in
    \cite[Theorem 1]{Akintunde+22}; however, its construction of
    step function \(\floor{\cdot}\) using \(\ReLU\) is incorrect
    as it is impossible to construct a discontinuous function
    from continuous functions using finitely many compositions,
    multiplications, and \(\pm\) operations.}
\end{proof}
\section{Future Research}
This paper aims to introduce the wealth of tools available in formal methods literature to the MAS domain. We have adapted several well-studied verification techniques to verify MAS safety with respect to \LTL{} specifications. Similar extensions of other logics, such as LTL[F] and LTL with discounting \cite{AlmagorBK16-Quality} or Alternating-Time Temporal Logic (ATL) and ATL* \cite{AlurHK02-ATL}, may be also considered. Similarly, there are numerous additional tools developed by the verification community that can be adapted for verifying MAS safety. Perhaps the most notable among these is IC3 \cite{Bradley11}. Other examples of literature that could potentially be exploited for verifying MAS safety, particularly to the problems considered here, are \cite{BradleyM08,Dillig+13,Padon+22}.

%%% Local Variables:
%%% mode: latex
%%% TeX-master: "../main"
%%% End:

\end{appendices}

%%%%%%%%%%%%%%%%%%%%%%%%%%%%%%%%%%%%%%%%%%%%%%%%%%%%%%%%%%%%%%%%%%%%%%%%

\end{document}